\documentclass[11pt]{article}

\usepackage{amsfonts}
\usepackage{amsmath}
\usepackage{amssymb}
\usepackage{amstext}
\usepackage{amsthm}
\usepackage{fullpage}
\usepackage{color}
\usepackage{eucal}
\usepackage{graphicx}
\usepackage[colorlinks=true,breaklinks,citecolor=blue,linkcolor=blue]{hyperref}
\usepackage{dsfont}

\newtheorem{theorem}{Theorem}
\newtheorem*{theorem*}{Theorem}
\newtheorem{lemma}[theorem]{Lemma}

\DeclareMathOperator*{\Exp}{\mathbb{E}}

\newcommand{\one}{{\mathds 1}}

\newcommand{\remove}[1]{}

\newcommand{\F}{\mathbb{F}}
\newcommand{\R}{\mathbb{R}}

\newcommand{\tf}{\widehat{f}}

\title{An Entropic Proof of Chang's Inequality}
\author{
Russell Impagliazzo\thanks{\texttt{russell@cs.ucsd.edu}, Department of Computer Science and Engineering, University of California San Diego} \and 
Cristopher Moore\thanks{\texttt{moore@santafe.edu}, Department of Computer Science, University of New Mexico and Santa Fe Institute} \and Alexander Russell\thanks{\texttt{acr@cse.uconn.edu}, Department of Computer Science and Engineering, University of Connecticut}}

\begin{document}
\maketitle

\begin{abstract}
Chang's lemma is a useful tool in additive combinatorics and the analysis of Boolean functions.  Here we give an elementary proof using entropy.  The constant we obtain is tight, and we give a slight improvement in the case where the variables are highly biased.
\end{abstract}

\section{The lemma}

For $S \in \{0,1\}^n$, let $\chi_k: \{ \pm 1 \}^n \to \R$ denote the character 
\[
\chi_S(x) = \prod_{i \in S} x_i \, . 
\]
For any function $f:\{ \pm 1 \}^n \to \R$, we can then define its Fourier transform $\tf:\{0,1\}^n \to \R$ as 
\[
\tf(S) = \Exp_x f(x) \chi_S(x) = \frac{1}{2^n} \sum_x f(x) \chi_S(x) \, . 
\]
For characters of Hamming weight $1$, we will abuse notation by writing $\tf(i)$ instead of $\tf(\{i\})$. 

Chang's lemma~\cite{talagrand,chang} places an upper bound on the total Fourier weight, i.e., the sum of $\tf^2$, of the characteristic function of a small set on the characters with Hamming weight one.

\begin{lemma}
\label{lem:chang}
Let $A \subseteq \{ \pm 1 \}^n$ such that $|A| = 2^n \alpha$, and let $f=\one_A$ be its characteristic function.  Then
\[
\sum_{i=1}^n \tf(i)^2 \le 2 \alpha^2 \ln \frac{1}{\alpha} \, . 
\]
\end{lemma}

\begin{proof}
Suppose that we sample $x$ according to the uniform distribution on $A$.  Since the mutual information is nonnegative, the entropy $H(x)$ is at most the sum of the entropies of the individual bits, 
\[
H(x) \le \sum_{i=1}^n H(x_i) \, .
\]
This gives
\begin{equation}
\label{eq:ineq}
n \ln 2 + \ln \alpha 
\le \sum_{i=1}^n h(p^+_i) 
\end{equation}
where $p^+_i$ denotes the probability that $x_i=+1$, 
\[
p^+_i 
= \frac{1}{2}\left( 1+\Exp_{x \in A} x_i \right) 
= \frac{1}{2}\left( 1+\frac{\tf(i)}{\alpha} \right) \, . 
\]
and where $h$ denotes the entropy function 
\[
h(p) = -p \ln p - (1-p) \ln (1-p) \, . 
\]
The Taylor series around $p=1/2$ gives
\begin{equation}
\label{eq:taylor}
h\!\left( \frac{1+x}{2} \right) = \ln 2 \;-\!\! \sum_{t=2, 4, 6, \ldots} \frac{x^t}{t(t-1)} \le \ln 2 - \frac{x^2}{2}  \, , 
\end{equation}
so~\eqref{eq:ineq} becomes
\[
\ln \alpha \le - \frac{1}{2} \sum_{i=1}^n \frac{\tf(i)^2}{\alpha^2} \, , 
\]
Rearranging completes the proof.
\end{proof}

\section{Variations}

The lemma (and our proof) apply equally well to the Fourier weight $\sum_{S \in B} \tf(S)^2$ of any basis $B$ of $\F_2^n$, since the set of parities $\{\prod_{i \in S} x_i \mid S \in B\}$ determines $x$.  
This gives the following commonly-quoted form of Chang's lemma.

\begin{lemma}
Let $A \subseteq \{ \pm 1 \}^n$ such that $|A| = 2^n \alpha$, and let $f=\one_A$ be its characteristic function.  Fix $\rho > 0$ and let $R \subset \F_2^n$ be the set $\{S : |\tf(S)| > \rho \alpha \}$.  Then $R$ spans a space of dimension less than $d = 2 \rho^{-2} \ln (1/\alpha)$.
\end{lemma}

\begin{proof}
If $R$ spans a space of dimension $d$ or greater, there is a set of $d$ linearly independent vectors in $R$.  Completing to form a basis $B$ gives $\sum_{S \in B} \tf(S)^2 > 2 \alpha^2 \ln (1/\alpha)$, violating Lemma~\ref{lem:chang}.
\end{proof}

For any integer $k \ge 1$, there are bases consisting entirely of vectors of Hamming weight $k$.  Fixing $k$ and averaging over all such bases gives
\[
\sum_{S: |S|=k} \tf(S)^2 
\le \frac{2}{n} {n \choose k} \,\alpha^2 \ln \frac{1}{\alpha}
\le \frac{2 n^{k-1}}{k!} \,\alpha^2 \log (1/\alpha) \big) \, . 
\]
This also follows immediately from Shearer's lemma.  However, this is noticeably weaker than the ``weight $k$ bound''
\[
\sum_{S: |S|=k} \tf(S)^2 = O\big( \alpha^2 \log^k (1/\alpha) \big) \, .
\]

Finally, we note that if some bits are highly biased, i.e., if $|\tf(i)|/\alpha$ is close to $1$, 
we can replace~\eqref{eq:taylor} with the bound 
\begin{equation}
\label{eq:tight-biased}
h(p) \le p(1- \ln p) \, ,
\end{equation}
which is tight when $p$ is small.  Combining this with the corresponding bound for $p$ close to $1$ gives
\[
h\!\left( \frac{1+x}{2} \right) \le \frac{1-|x|}{2} \left( 1 - \ln \frac{1-|x|}{2} \right) \, .
\]
We compare this bound with~\eqref{eq:taylor} in Figure~\ref{fig:entropy-bounds}.  This gives another version of Lemma~\ref{lem:chang}:

\begin{lemma}
\label{lem:chang-tight}
Let $A \subseteq \{ \pm 1 \}^n$, let $f=\one_A$ be its characteristic function, and let 
\[
\delta_i=\frac{1}{2} \left( 1-\frac{|\tf(i)|}{\alpha} \right) = \min\left( p_i^+, 1-p_i^+ \right) \, . 
\]
Then
\begin{equation}
\label{eq:chang-biased}
\sum_{i=1}^n \delta_i \left( 1-\ln \delta_i \right) \ge \ln |A| \, . 
\end{equation}
\end{lemma}
\noindent
This is nearly tight, for instance, if $A$ is the set of vectors with Hamming weight $1$.  Then $|A|=n$, $\delta_i=1/n$, and~\eqref{eq:chang-biased} reads $1+\ln n \ge \ln n$.  

\begin{figure}
\begin{center}
\includegraphics[width=4in]{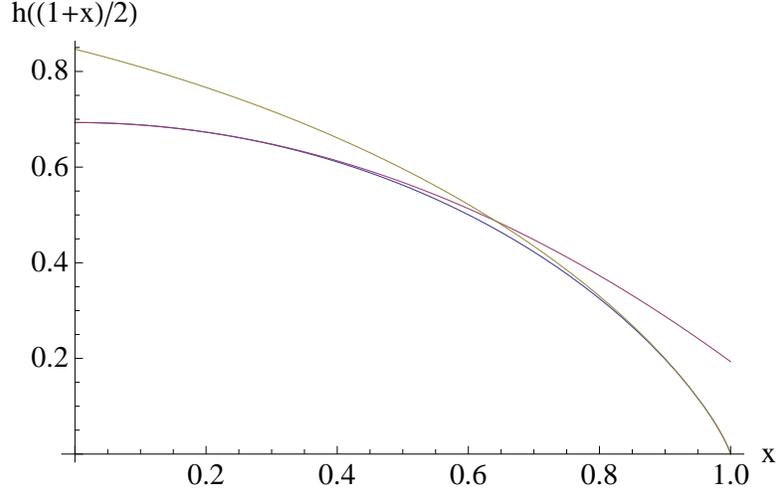}
\end{center}
\caption{The entropy function $h(p)$ where $p=(1+x)/2$ and $x \le 0 \le 1$, with the upper bounds~\eqref{eq:taylor} (which is tight when $|x|$ is small) and~\eqref{eq:tight-biased} (which is tight when $|x|$ is close to 1).}
\label{fig:entropy-bounds}
\end{figure}

\section*{Acknowledgments}  
We thank Ryan O'Donnell for a wonderful set of lectures on the analysis of Boolean functions at the Bellairs Research Institute, and Ran Raz for helpful communications.  C.M. and A.R. are supported by NSF grant CCF-1117426 and ARO contract W911NF-04-R-0009.

\end{document}